\documentclass[reprint,aps,pra,amsmath,amssymb,superscriptaddress]{revtex4-2}

\usepackage{graphicx}
\usepackage{bm}
\usepackage{amsthm}
\usepackage[colorlinks,allcolors=blue]{hyperref}

\newtheorem{theorem}{Theorem}

\theoremstyle{definition}

\makeatletter
\let\orig@section\section
\renewcommand\section{\par\global\@nobreakfalse\orig@section}
\makeatother

\newcommand{\ket}[1]{\lvert #1 \rangle}
\newcommand{\bra}[1]{\langle #1 \rvert}

\begin{document}

\title{A three-eighths upper bound on the stabilizer-rank exponent of face-center magic states}

\author{Samuel Bevins}
\email{sjbevins@wm.edu}
\affiliation{
    Department of Physics,
    The College of William \& Mary,
    Williamsburg, VA 23187, USA
}

\date{\today}

\begin{abstract}
We give an explicit construction proving that the exact stabilizer rank of $m$ copies of the qubit face-center magic state is $O(2^{3m/8})$. Two symmetry-related fifteen-term identities combine into a twenty-four-term decomposition at twelve copies. The balanced structure of this decomposition yields a substitution that adds eight copies at a cost of eight terms per existing term. The resulting bound $\gamma_T \leq 3/8$ improves the previously published exponent $\log_2(3)/4$. 
\end{abstract}

\maketitle

\section{Introduction and Result}

Stabilizer decompositions express magic states as complex linear combinations of stabilizer states. Their size is an important parameter in classical simulation methods for circuits with nonstabilizer resources \cite{Bravyi2019, Qassim2021}. We follow the qubit orbit notation of Labib and Russo \cite{labib2026}. For a pure state $\ket{\psi}$, its exact complex stabilizer rank $\chi(\ket{\psi})$ is the smallest number of pure stabilizer states whose span contains $\ket{\psi}$. 

The qubit $T$-type representative is given by
\begin{equation}
    \ket{T} = \frac{1}{\sqrt{6+2\sqrt{3}}} \begin{bmatrix} 1 + \sqrt{3} \\ 1 + i \end{bmatrix}\,.\label{eq:1}
\end{equation}
Its Bloch vector is $(1,1,1)/\sqrt{3}$, at the center of a face of the stabilizer octahedron. Following Bravyi and Kitaev \cite{BravyiKitaev2005}, these states form the $T$-type orbit. The conventional equatorial $T$-state belongs to the $H$-type orbit in this convention. Qassim, Pashayan, and Gosset instead denote the face-center state by $\ket{F}$ \cite{Qassim2021}. 

Define
\begin{equation}
    \gamma_T = \lim_{m\to\infty}\frac{1}{m} \log_2 \chi(\ket{T}^{\otimes m})\,.\label{eq:2}
\end{equation}
Since stabilizer rank is submultiplicative under tensor products, the sequence $\log_2\chi(\ket{T}^{\otimes m})$ is subadditive. Fekete's lemma therefore establishes that this limit exists. The previously published bound for this orbit is $\gamma_T \leq \log_2(3)/4 \simeq 0.39624$ \cite{Qassim2021, labib2026}.
These are the $T$-type magic states distilled in the Bravyi--Kitaev protocol~\cite{BravyiKitaev2005}, so the cost of simulating circuits that consume $m$ of them scales with $\chi(\ket{T}^{\otimes m})$.

\begin{theorem}
    For every integer $j \geq 0$, there is an explicit decomposition of $\ket{T}^{\otimes (4+8j)}$ into exactly $3\cdot 8^j$ distinct, nonzero stabilizer summands. Consequently, for every integer $m\geq 1$,
    \begin{equation}
        \chi(\ket{T}^{\otimes m}) \leq 3\cdot 8^{\lceil(m-4)/8\rceil}= O(2^{3m/8})\,, \quad \gamma_T \leq \frac{3}{8}\,.\label{eq:chi}
    \end{equation}
\end{theorem}

The improvement uses two ideas. First, a dictionary symmetry glues two fifteen-term identities into a twenty-four-term witness at twelve copies. Repeating that witness alone would give the exponent $\log_2(24)/12\simeq 0.38208$. Second, the witness splits into three eight-term groups, yielding a substitution with the better exponent $\log_2(8)/8=3/8$. Neither step assumes that the constructed decompositions are optimal.

\section{Three-letter Stabilizer Dictionary}

Let $\{a,b\}$ be a basis of $V = \mathbb{C}^2$ and let
\begin{equation}
    c=a+b\,,\quad z=\frac{1+i\sqrt{3}}{2}\,,\quad v=a+zb\,.\label{eq:4}
\end{equation}
We will use italic letters $a$, $b$, $c$ for vectors in $V$ and typewriter letters $\texttt{a}$, $\texttt{b}$, $\texttt{c}$ for the corresponding one-letter words. A word such as $\texttt{aabba}$ represents the ordered tensor product $a\otimes a\otimes b \otimes b \otimes a$. 
We will also use $z^2=z-1$, $z^3=-1$, and $z(1-z)=1$. Let $R_k$ be the smallest number of nonzero words in an expansion of $v^{\otimes k}$. 

Define the map $L:\, V\to (\mathbb{C}^2)^{\otimes 2}$ by the mapping
\begin{align}
    L(a) &= u_0 = (i,i,i,-i)^{\rm T}\,,\nonumber \\ L(b) &= u_1 = (1-i,0,0,1+i)^{\rm T}\,, \label{eq:L}\\ L(c) &= u_2 = (1,i,i,1)^{\rm T} = u_0+u_1\,. \nonumber
\end{align}
The respective pairs of positive Pauli stabilizers are $(ZX,XZ)$, $(ZZ,XY)$, and $(XX,ZY)$, with juxtaposition denoting tensor products. Let $f = (1+\sqrt{3}, 1+i)^{\rm T}$. Direct multiplication yields
\begin{equation}
    f^{\otimes 2} = \kappa L(v)\,,\quad \kappa = (1+\sqrt{3})(1-i)\,.\label{eq:f2}
\end{equation}
Since $f$ is proportional to $\ket{T}$, every $r$-word expression for $v^{\otimes k}$ gives a stabilizer decomposition with at most $r$ terms on $2k$ qubits. Thus $\chi(\ket{T}^{\otimes 2k})\leq R_k$. The map $L$ is injective, and distinct dictionary words produce distinct physical stabilizer rays. 

Our starting point is the three-word identity
\begin{equation}
    v^{\otimes 2} = (1-z)\texttt{aa}-\texttt{bb} + z\texttt{cc}\,,\quad R_2\leq 3\,.\label{eq:v2}
\end{equation}
It reproduces the optimal four-copy value $\chi(\ket{T}^{\otimes 4})=3$~\cite{labib2026}.

\section{Symmetry Gluing}
The construction begins with the exact identity
\begin{equation}
    \begin{aligned}
        v^{\otimes 5} = &-(1+z)\texttt{aaaaa} + (2z-1)\texttt{bbbbb} + (z-2)\texttt{ccccc} \\ &+z(\texttt{aabba} + \texttt{bbaaa} - \texttt{abbab} - \texttt{baabb}) \\ &+\texttt{acaca} + \texttt{cacaa} + \texttt{accac} +\texttt{caacc} \\ &+(1-z) (\texttt{bbccc} + \texttt{ccbbc} + \texttt{bcbcb} + \texttt{cbcbb})\,.
    \end{aligned}\label{eq:exactidentity}
\end{equation}
Appendix A gives a finite coordinate proof. Write $\alpha_w$ for the coefficient of word $w$ and take $\alpha_w =0$ outside the listed support. 

Define the linear map $\tau:\, V \to V$ by
\begin{equation}\tau(a) = b\,,\quad \tau(b) = -c\,,\quad \tau(c) = -a\,.
\end{equation} 
This map satisfies $\tau^3=I$ and $\tau(v) = -zv$. Therefore,
\begin{equation}
    S =\tau\otimes \tau \otimes \tau \otimes \tau \otimes \tau^2
\end{equation}
fixes $v^{\otimes 5}$. Acting with $S$ on Eq. \eqref{eq:exactidentity} yields a second fifteen-word expansion with coefficients $\beta_w$. 

The two supports overlap in exactly six words. On each shared word, $\alpha_w=z\beta_w$.
\begin{equation}
    \begin{array}{c|cc}
        \text{Words} & \alpha_w & \beta_w  \\\hline \texttt{abbab}, \texttt{baabb} & -z & -1 \\ \texttt{acaca}, \texttt{cacaa} & 1 & 1-z \\ \texttt{bbccc}, \texttt{ccbbc} & 1-z & -z
    \end{array}\label{eq:sixwords}
\end{equation}
Append $a$ to each word of the first representation, and $zb$ to the second. For any shared word $w$,
\begin{equation}
    \alpha_w(wa) + z\beta_w(wb) = \alpha_w(wc)\,.
\end{equation}
The equality holds because $c=a+b$. Each match replaces two terms by one. We obtain
\begin{equation}
    v^{\otimes 6} = \sum_{w\in W} d_ww\,,\quad |W| = 15+15-6=24\,.\label{eq:v6}
\end{equation}
\begin{table}[t]
\caption{The exact twenty-four-word identity $v^{\otimes6}=\sum_{w\in W}d_ww$. Each first-letter group has eight words, and all twenty-four five-letter suffixes are distinct.}
\label{tab:witness}
\centering
\scriptsize
\setlength{\tabcolsep}{2.5pt}
\begin{ruledtabular}
\begin{tabular}{ll@{\hspace{4pt}}ll@{\hspace{4pt}}ll}
$a$-group & $d_w$ & $b$-group & $d_w$ & $c$-group & $d_w$\\
$\texttt{aaaaaa}$&$-(1+z)$&$\texttt{baabbc}$&$-z$&$\texttt{caacca}$&$1$\\
$\texttt{aaaabb}$&$1+z$&$\texttt{babacb}$&$-z$&$\texttt{cacaac}$&$1$\\
$\texttt{aabbaa}$&$z$&$\texttt{bbaaaa}$&$z$&$\texttt{cbbcab}$&$1-z$\\
$\texttt{aaccbb}$&$-1$&$\texttt{bbbbba}$&$2z-1$&$\texttt{cbcbba}$&$1-z$\\
$\texttt{ababcb}$&$-z$&$\texttt{bbbbcb}$&$2z-1$&$\texttt{ccaabb}$&$-1$\\
$\texttt{abbabc}$&$-z$&$\texttt{bbcccc}$&$1-z$&$\texttt{ccbbcc}$&$1-z$\\
$\texttt{acacac}$&$1$&$\texttt{bcbcba}$&$1-z$&$\texttt{ccccab}$&$z-2$\\
$\texttt{accaca}$&$1$&$\texttt{bccbab}$&$1-z$&$\texttt{ccccca}$&$z-2$
\end{tabular}
\end{ruledtabular}
\end{table}
Table \ref{tab:witness} contains the full identity. Exactly eight words begin with each of $a$, $b$, $c$. All twenty-four suffixes obtained by deleting the first letter are distinct. This property drives the recursion.

\section{An Eight-Branch Substitution}
Group Eq. \eqref{eq:v6} by its first letter,
\begin{equation}
    v^{\otimes 6} = a\otimes A + b\otimes B + c\otimes C\,,
\end{equation}
where $A,B,C \in V^{\otimes 5}$ each contain eight words. Write $U = v^{\otimes 5}$. Comparing the coefficients of $a$ and $b$ in the first slot gives the tensor identities
\begin{equation}
    A + C = U\,,\quad B+C = zU\,.\label{eq:tensoridentities}
\end{equation}
These are equalities of tensors, not of coefficient lists in the linearly dependent dictionary. 

Define the map $E:\, V\to V^{\otimes 5}$ by
\begin{equation}
    E(a) = C\,,\quad E(b) = (1-z)A\,,\quad E(c) = -zB\,.\label{eq:16}
\end{equation}
Using the relations in Eq. \eqref{eq:tensoridentities} together with $z^2=z-1$,
\begin{equation}
    \begin{aligned}
         E(a) + E(b) & = -zB =E(c)\,, \\ E(v) &= C+z(1-z)A = U\,.
    \end{aligned}\label{eq:E}
\end{equation}
Therefore, we see that $E$ is linear and $E(v) = v^{\otimes 5}$. Each letter expands into eight nonzero words. Their supports are disjoint because the first-letter groups of Table \ref{tab:witness} have pairwise distinct suffixes. 

Applying $E$ to the first slot of an $r$-word decomposition of $v^{\otimes k}$ produces $8r$ words for $v^{\otimes (k+4)}$. Words with different input tails retain different tails. For a shared tail, distinct input letters have disjoint output supports. No collisions or cancellations occur, and hence
\begin{equation}
    R_{k+4}\leq 8R_k\,.\label{Rineq}
\end{equation}
\begin{proof}[Proof of Theorem 1.]
    Starting with the three-word decomposition in Eq. \eqref{eq:v2}, apply the substitution $E$ from Eq. \eqref{eq:16} to the first slot $j$ times. Each application replaces every word with eight distinct nonzero words without collisions. This constructs $3\cdot 8^j$ distinct, nonzero words for $v^{\otimes (2+4j)}$. Equations \eqref{eq:L} and \eqref{eq:f2} map these terms to physical stabilizer states on $4+8j$ qubits. 

    Postselecting a qubit of a stabilizer state onto $\bra{0}$ gives zero or an unnormalized stabilizer state \cite{gottesman1997stabilizer}, 
    so the number of terms cannot increase. Since $\langle0|T\rangle \neq 0$, postselecting $4+8j-m$ copies leaves a nonzero multiple of $\ket{T}^{\otimes m}$.
    Choose $j = \lceil (m-4)/8\rceil$ to obtain Eq. \eqref{eq:chi} for every $m\geq 1$. In addition, $j < (m-4)/8 +1$ gives a constant prefactor independent of $m$, establishing $\gamma_T \leq 3/8$.  
\end{proof}

\section{Discussion}
The result reduces the known constructive upper bound on $\gamma_T$ from $\log_2(3)/4\simeq 0.39624$ to $3/8=0.375$. It does not determine the optimal exponent, establish a lower bound, or guarantee a particular simulation runtime. We also do not separate $T$-type and $H$-type orbit exponents. Whether symmetry gluing improves bounds for other magic-state orbits remains open.

\begin{acknowledgments}
The author thanks Yunus Bidav for helpful discussions, particularly regarding the design and development of many of the AI tools used in this work.
\end{acknowledgments}

\section*{Author contributions and use of generative AI}
The author directed this work, verified all results, wrote the final manuscript, and takes full responsibility for its content. The results emerged primarily from AI-assisted research conducted using Claude Opus 5.5 (Anthropic) and Qwen3.8-27b (Alibaba Cloud) running locally. AI systems played a central role in discovering the algebraic identities, developing the constructions and proofs, and producing computational verification code, which is available upon request. AI tools also assisted with drafting and editing the text.

\appendix

\section{Coordinate check of the fifteen-word identity}
Write the right side of Eq. \eqref{eq:exactidentity} as $P\otimes a + Q\otimes b + H\otimes c$, with
\begin{equation}
\begin{aligned}
    P&=-(1+z)\texttt{aaaa} + z(\texttt{aabb} +\texttt{bbaa}) + \texttt{acac} + \texttt{caca} \,,  \\ Q&= (2z-1)\texttt{bbbb} - z(\texttt{abba} + \texttt{baab}) \\ &\quad \,\, +(1-z)(\texttt{bcbc} + \texttt{cbcb})\,, \\ H&= (z-2)\texttt{cccc} + \texttt{acca} + \texttt{caac} \\ &\quad \,\, +(1-z)(\texttt{bbcc} + \texttt{ccbb})\,. 
    \end{aligned}
\end{equation}
\begin{table}[t]
\caption{Coordinate check of Eq.~\eqref{eq:exactidentity}. The seven orbits contain all sixteen basis words.}
\label{tab:orbit}
\centering
\begin{ruledtabular}
\begin{tabular}{lrrrr}
Basis word & Size & $P$ & $Q$ & $H$\\
$\texttt{aaaa}$&1&$1-z$&$0$&$z$\\
$\texttt{aaab}$&4&$1$&$0$&$z-1$\\
$\texttt{aabb}$&2&$z$&$0$&$-1$\\
$\texttt{abab}$&2&$1$&$1-z$&$z-2$\\
$\texttt{abba}$&2&$0$&$-z$&$z-1$\\
$\texttt{abbb}$&4&$0$&$1-z$&$-1$\\
$\texttt{bbbb}$&1&$0$&$1$&$-z$
\end{tabular}
\end{ruledtabular}
\end{table}
Because $c=a+b$, Eq. \eqref{eq:exactidentity} follows from $P+H = v^{\otimes 4}$ and $Q+H = zv^{\otimes 4}$. The three tensors are invariant under $(12)(34)$ and $(13)(24)$, so it suffices to check the seven orbit representatives in Table II. 
\\
\section{The twenty-four-word witness}
Table \ref{tab:witness} gives the coefficients in Eq. \eqref{eq:v6}. To reconstruct it, apply $S$ to Eq. \eqref{eq:exactidentity}, append $a$ to the nine unmatched words of Eq. \eqref{eq:exactidentity}, and append $b$ with an extra factor $z$ to the nine unmatched words of its image. Append $c$ to the six matched words. Grouping by first letter then determines $A$, $B$, $C$, and the substitutions.

\bibliography{references}

\end{document}